\documentclass[]{article}

\usepackage[T2A,T1]{fontenc}
\usepackage[utf8]{inputenc}
\usepackage[russian,english]{babel}
\usepackage{graphicx}
\usepackage{booktabs}
\usepackage{amssymb}
\usepackage{amsthm}
\usepackage{amsmath}
\usepackage{algorithm}
\usepackage[noend]{algpseudocode}
\usepackage{xcolor}
\usepackage{listings}
\usepackage{pifont}
\usepackage{array}
\usepackage{color,soul}

\newcommand{\cmark}{\ding{51}}
\newcommand{\xmark}{\ding{55}}

\theoremstyle{definition}
\newtheorem{definition}{Definition}%[section]
\newtheorem{theorem}{Theorem}

\definecolor{mGreen}{rgb}{0,0.6,0}
\definecolor{mGray}{rgb}{0.5,0.5,0.5}
\definecolor{mPurple}{rgb}{0.58,0,0.82}
\definecolor{backgroundColour}{rgb}{0.95,0.95,0.92}

\lstdefinestyle{CStyle}{
	backgroundcolor=\color{backgroundColour},   
	commentstyle=\color{mGreen},
	keywordstyle=\color{magenta},
	numberstyle=\tiny\color{mGray},
	stringstyle=\color{mPurple},
	basicstyle=\footnotesize,
	breakatwhitespace=false,         
	breaklines=true,                 
	captionpos=b,                    
	keepspaces=true,                 
	showspaces=false,                
	showstringspaces=false,
	showtabs=false,                  
	tabsize=2,
	language=C
}

\makeatletter
\def\BState{\State\hskip-\ALG@thistlm}
\makeatother

\title{On Function Description}
\author{{Rade Vuckovac 
		(\foreignlanguage{russian}{Раде Вучковац})
		}}

\begin{document}

\maketitle

\begin{abstract}
The main result is that: function descriptions are not made equal, and they can be categorised in at least two categories using various computational methods for function evaluation. The result affects Kolmogorov complexity and Random Oracle Model notions. More precisely, the idea that the size of an object and the size of the smallest computer program defining that object is a ratio that represents the object complexity needs additional definitions to hold its original assertions.

\end{abstract}

\section{Introduction}

An introduction to Kolmogorov complexity could start with Chaitin's short story~\cite{chaitin1975randomness}:
\begin{quotation}
	Suppose you have a friend who is visiting a planet in
	another galaxy, and that sending him telegrams is very expensive. He forgot to take along his tables of
	trigonometric functions, and he has asked you to supply them. You could simply translate the numbers into an
	appropriate code (such as the binary numbers) and transmit them directly, but even the most modest tables of
	the six functions have a few thousand digits, so that the cost would be high. A much cheaper way to convey
	the same information would be to transmit instructions for calculating the tables from the underlying
	trigonometric formulas, such as Euler's equation $ e^{ix} = cos x + i sin x $. Such a message could be relatively brief, yet inherent in it is all the information contained in even the largest tables.
\end{quotation}

\begin{figure}
	\centering
	\includegraphics[width=0.4\linewidth]{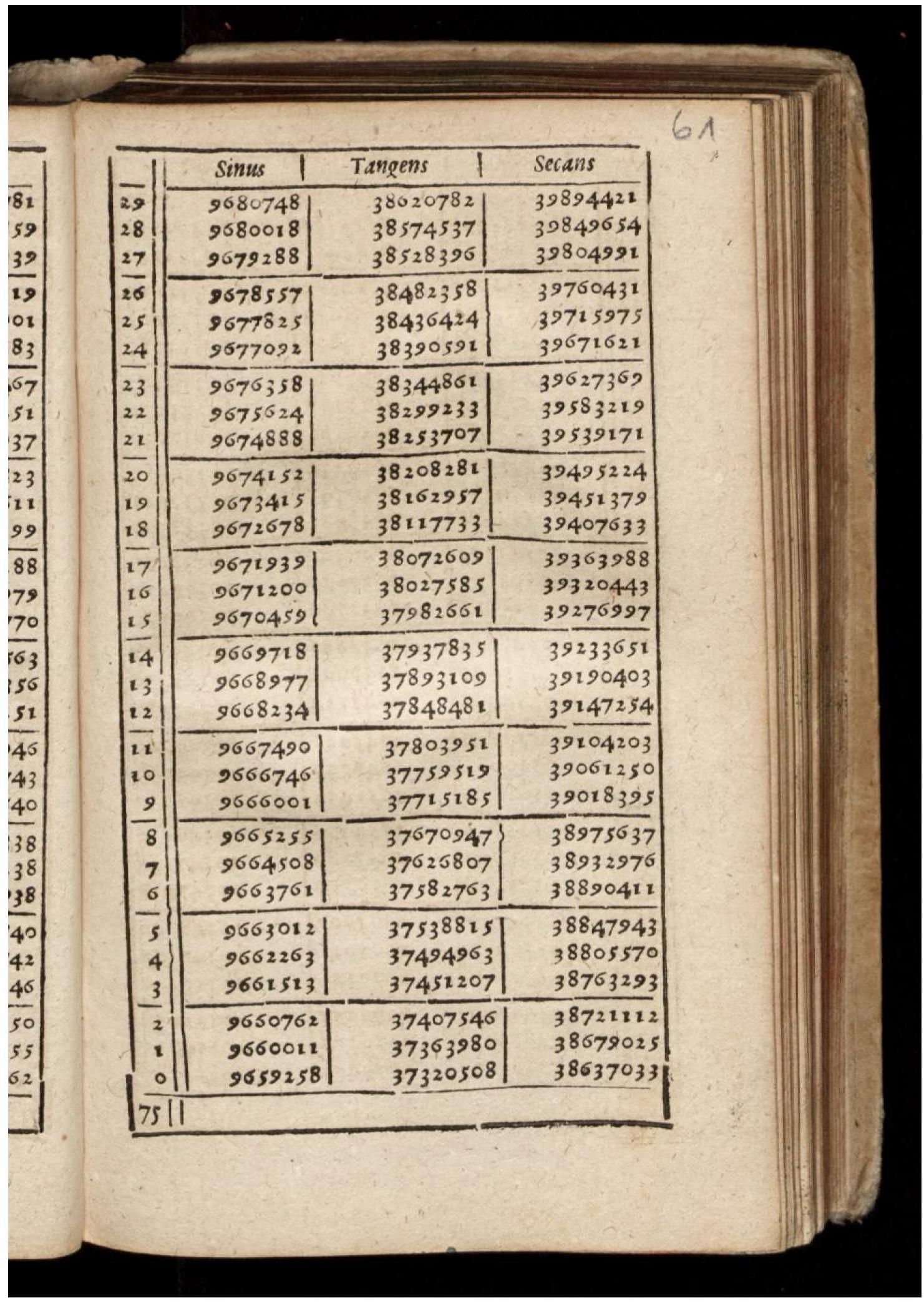}
	\caption{\small \sl Trigonometry book from 17th century; Matthias Bernegger - Manuale Mathematicum darinn begriffen/ Die Tabulae Sinuum, Tangentium, Secantium.}
	\label{fig:berneggermanuale137}
\end{figure}

The Kolmogorov complexity deals with objects and computer programs outputting it. If the computer program is more compact than the object as an output, the complexity of the object in question is considered less complicated. The Kolmogorov complexity can apply to any object. In our case, it will deal with mathematical functions. Function description is a computer program which will define an object; in our case domain, co-domain and their mappings.

Figure~\ref{fig:berneggermanuale137} shows the tables of some trigonometric functions and are represented as a string $ s $. Implementation of Euler's equation is an example of trigonometric functions description and also a description of the string $ d(s) $. Then Kolmogorov complexity $ K $ of string $ s $ is stated as:
\begin{equation}
K(s)=|d(s)|
\end{equation}
where $ |d(s)| $ is a length of a program, in our case the length of the program implementing Euler's equation.
There is also relation between $ K(s) $ and the length of a string $ |s| $:
\begin{equation}
\forall s.K(s)\leq|s| + c
\end{equation}
where $ c $ is a constant. From Chaitin's short story it becomes obvious that Euler's equation is a way shorter than Figure~\ref{fig:berneggermanuale137}:
\begin{equation}
|d(s)|<|s| + c
\end{equation}
\\
To some extent, the concept of Kolmogorov complexity was used for argumentation in the negative result of the random oracle model.
A random oracle is a black box which on a given inquiry, replies with a random
answer. A fictional story explains the concept quickly. In the black box lives
a gnome with some dice and a blank notebook. Anyone can submit a question (an
input $ q $) to the black box. When $ q $ is submitted, the gnome checks whether the input $ q $ is already in
the notebook. If it is there the gnome will respond with result $ r $ from the notebook. If $ q $ is not in the notebook, the gnome will throw its dice and the result $ r $ will be recorded to the notebook as a mapping from the query. That result $ r $ returns to the submitter. The notebook entries may look like Table~\ref{tab:randomOracle}
where $ q $ and $ r $ are sorted for easier searching.
\begin{table}[h]
	\begin{center}
		\begin{tabular}{cc|cc|cc}
			\toprule
			$q$ & $r$ & sort by $q$ & r & q &sort by $r$  \\ 
			\midrule
			 4	&	90	&	2	&	56	&	66	&	22	\\
			 27	&	35	&	4	&	90	&	27	&	35	\\
			 2	&	56	&	19	&	54	&	62	&	39	\\
			 19	&	54	&	20	&	89	&	67	&	48	\\
			 96	&	93	&	27	&	35	&	19	&	54	\\
			 98	&	99	&	62	&	39	&	2	&	56	\\
			 67	&	48	&	66	&	22	&	20	&	89	\\
			 66	&	22	&	67	&	48	&	4	&	90	\\
			 20	&	89	&	96	&	93	&	96	&	93	\\
			 62	&	39	&	98	&	99	&	98	&	99	\\
			
			\bottomrule
		\end{tabular}
	\end{center}
	\caption{\small \sl Random oracle entries and its sorted mappings.}
	\label{tab:randomOracle}
\end{table}

The availability of a random oracle serves as a security assumption during the
design and development of various cryptographic protocols. For example, most
of RSA (a well known public-key cryptography method) signing and encryption
are shown to be secure under the Random Oracle Model (ROM)~\cite{bellare1993random}.
The problem starts when we wish to implement the fictional random oracle
with a real algorithm on modern computers. The paper "The Random
Oracle Methodology, Revisited"(ROMR)~\cite{canetti2004random} shows that a random oracle is unreplaceable
by any function (generally hash function) and that proofs based on ROM are unsound. Informal arguments are:

\begin{quotation}
	An obvious failure. We first comment that an obvious maximalistic definition, which amount
	to adopting the pseudorandom requirement of~\cite{goldreich1986construct}, fails poorly. That is, we cannot require that
	an (efficient) algorithm that is given the description of the function cannot distinguish its inputoutput behavior from the one of a random function, because the function description determines 	its input-output behavior.
\end{quotation}

\begin{quotation}
	Informal Theorem 1.1 There exist no correlation intractable function ensembles.
... The proof of the above negative result relies on the fact
that the description of the function is shorter than its input.
\end{quotation}

Resolving the random oracle is a significant and challenging endeavour. Shai's (one of the authors~\cite{canetti2004random}) view is:
\begin{quotation}
	Another possible explanation is that the random oracle methodology works for the currently published schemes, due to some specific features of these schemes that we are yet to identify. That is, maybe it is possible to identify interesting classes of schemes, for which security in the Random Oracle Model implies the existence of a secure implementation. Identifying such interesting classes, and proving the above implication, is an important and seemingly hard research direction.
\end{quotation}

This work sheds light on the random oracle controversy by defining function descriptions. It puts a function description of a program into two categories. Each category has a different rank depending on what computational methods are available to each category.

\section{Function Description, a Mathematical Approach}
The trigonometrical functions have been major engineering tools for a long, long time. Figure~\ref{fig:berneggermanuale137}  shows pre-calculated values of various trigonometrical constructs including sine function. Figure~\ref{fig:sine} shows the sine description.

The $ sine $ function deals with angles of a  right triangle. The sine of such an angle $ \Theta $ is the ratio of the opposite side and hypotenuse.
\begin{equation}
	sin \Theta = \frac {Opposite}{Hypotenuse}
\end{equation}

\begin{figure}[h]
	\centering
	\includegraphics[width=0.7\linewidth]{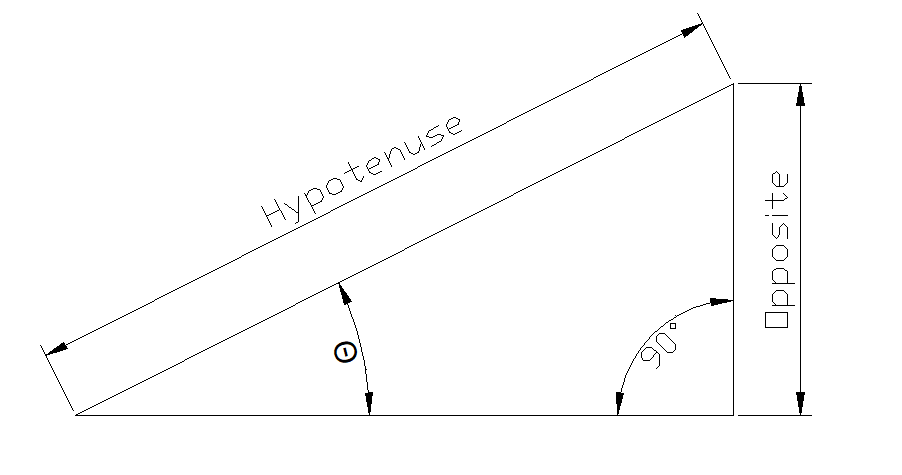}
	\caption{\small \sl Sine function; $ sin \Theta = \frac {Opposite}{Hypotenuse} $.}
	\label{fig:sine}
\end{figure}

\pagebreak

More formal definition via complex analysis is here:\\
http://us.metamath.org/mpegif/df-sin.html\\
From a computational view, an evaluation of $ sin $ function can be done in handful ways. Those ways we will call rankings.
\subsection{Rank 1}
\textbf{Rank 1;} \textit{It is an ability to compute an output providing a input via computer program.}
The example of Rank 1 is the C language $ sin $ implementation (given by Blindy:
https://stackoverflow.com/questions/2284860/how-does-c-compute-sin-and-other-math-functions)

\begin{equation}\label{taylor}
sin(x) := x - x^{3/3!} + x{^5/5!} - x^{7/7!} + ...
\end{equation}
 
Figure~\ref{lst:sine} shows an implementation.

\begin{figure}[h]
	\begin{lstlisting}[style=CStyle]
		float sin(float x)
		{
			float res=0, pow=x, fact=1;
			for(int i=0; i<5; ++i)
			{
				res+=pow/fact;
				pow*=-1*x*x;
				fact*=(2*(i+1))*(2*(i+1)+1);
			}
			return res;
		}
	\end{lstlisting}
	\caption{\small \sl $ sin $ implementation in C language; on $ 32 $-bit angle input, returns $ 32 $-bit decimal value.}
	\label{lst:sine}
\end{figure}
The description length $ |d(s)| $ is Figure~\ref{lst:sine} implementation size and it is in tens of KBytes region $ |d(s)|\approx20KB $. This particular $ sin $ function defines inputs and outputs as $ 32\,bit $ decimal values, then size of domain-range mapings is $ |s|=2*2^{32}\,bits $. Therefore we have an indication of low $ sin  $ complexity.
\begin{equation}\label{eq:description1}
|d(s)|<|s|
\end{equation}

\subsection{Rank 2}
\textbf{Rank 2;} \textit{A look-up table is used to evaluate input-output enquiries.}\\
If the function description has Rank 1, it automatically can obtain Rank 2 by creating a look-up table (for example Table~\ref{tab:sineTable}). 

For trigonometrical functions in game development, the look-up table is the preferred method. The reason is the spatial nature of some games and needs too many geometry calculations. Another advantage of the table is that they can be searched either by input or by output.

For example, the question is what angle ($ x $) produces $ 0.829 $ ($ sin(x)=0.829 $)? A searching algorithm will search for $ 0.829 $ through $ Sine $ column and will output $ 56 $ Degrees (Table~\ref{tab:sineTable}).

\begin{table}[h]
	\begin{center}
		\begin{tabular}{ccc|ccc}
			\toprule
			Index & Degrees & Sine & Index & Degree & Sine   \\ 
			\midrule
			1 & 45   & 0.7071	&	17 & 61	& 0.8746 \\ 
			2 & 46   & 0.7193   &	18 & 62	& 0.8829 \\               
			3 & 47   & 0.7314   &	19 & 63	& 0.8910 \\               
			4 & 48   & 0.7431   &	20 & 64	& 0.8988 \\               
			5 & 49   & 0.7547   &	21 & 65	& 0.9063 \\              
			6 & 50   & 0.7660   &	22 & 66	& 0.9135 \\               
			7 & 51   & 0.7771   &	23 & 67	& 0.9205 \\               
			8 & 52   & 0.7880   &	24 & 68	& 0.9272 \\              
			9 & 53   & 0.7986   &	25 & 69	& 0.9336 \\               
			10 & 54   & 0.8090   &	26 & 70	& 0.9397 \\               
			11 & 55   & 0.8192   &	27 & 71	& 0.9455 \\               
			12 & \textbf{56}   & \textbf{0.8290}   &	28 & 72	& 0.9511 \\               
			13 & 57   & 0.8387   &	29 & 73	& 0.9563 \\               
			14 & 58   & 0.8480   &	30 & 74	& 0.9613 \\               
			15 & 59   & 0.8572   &	31 & 75	& 0.9659 \\               
			16 & 60   & 0.8660   &	32 & 76	& 0.9703 \\              
			\bottomrule
		\end{tabular}
	\end{center}
	\caption{\small \sl Sine Table: Partial range from 45 to 76 degrees.}
	\label{tab:sineTable}
\end{table}

%In this case, the program implementing this approach needs to include domain-range mapping. The consequence is that program description size is approximately the same as look-up table mapping $ |d(s)|\approx|s| $. That is OK for the close to one degree of precision requirements. $ 32\,bit $ precision grain (float) will need $ 4294967296 $ entries to look-up table.

\subsection{Rank 3}
The quoted story from the beginning could continue with another twist. The communication with Earth and the faraway planet is lost. Your friend received function for calculating $ sin $, but he needs to calculate angle when the $ sin $ value is known. Essentially knowing $ y $, pair it with corresponding $ x $. With tables, that is an easy task, but they are forgotten. One way to solve that problem is to guess $ x $ and run $ sin(x) $ and see if the $ y $ in question match. That exercise might be expensive because that search is exhaustive. Since $ sin $ have mathematical structure, simple binary search saves the day.

For example, the question of what angle ($ x $) produces $ 0.829 $? The searching algorithm will:
\begin{list}{}{}
	\item 1st try $ sin(45)=0.707 $
	\item 2nd try $ \downarrow $ $ sin(30)=0.500 $
	\item 3rd try $ \uparrow $ $ sin(60)=0.866 $
	\item 4th try $ \downarrow $ $ sin(55)=0.819 $
	\item 5th try $ \uparrow $ $ sin(57)=0.839 $
	\item 6th try $ \downarrow $ $ sin(\textbf{56})=\textbf{0.829} $, \textbf{bingo!}
\end{list}

\textbf{Rank 3;} \textit{Ability to invert a function using the efficient search algorithm ($ O(log_2 n) $) without pre-calculated look-up table}.\\
If the function definition includes some input-output structure, to inverse that function the efficient search method suffices.
This method is universal because searching does not depend on a function description. The requirement is that a input-output pairs are in correlation, which is generally assumed, because function description is smaller than its input-output ($ |d(s)|<|s| $) (ROMR~\cite{canetti2004random}).\\ 
The efficiency of this approach is less than using Rank 4 function but is still efficient enough because the binary searching cost is logarithmic by nature ($ O(log_2 n) $).

\subsection{Rank 4}
\textbf{Rank 4;} \textit{Using inverse function to match output to the corresponding input}.\\
Instead using Rank 2 or 3 approach you can use calculator because almost every scientific calculator has $ sin^{-1} $ function which is inverse of $ sin $ (See Figure~\ref{fig:minussin} for example).

\begin{figure}[h]
	\centering
	\includegraphics[width=0.7\linewidth]{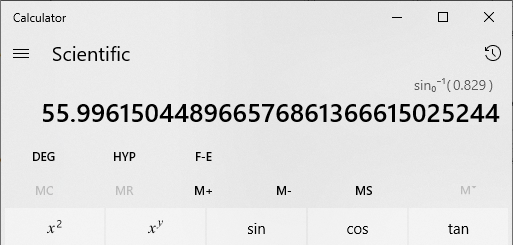}
	\caption{Win 10 Calculator; calculating $ sin^{-1}(0.829) $}
	\label{fig:minussin}
\end{figure}

The details of $ sin^{-1} $ function could be found here:\\
http://mathworld.wolfram.com/InverseSine.html

\pagebreak
 
\section{Function Description, the Same but Different}\label{sec:tsp}

Travelling salesman problem (TSP) is one of complete $ NP $ problems. $ NP $ problems are solved by exhaustive search, and the solution is verified efficiently (polynomial cost). "Complete" qualifier means that if one of $ NP $ complete problems are solved efficiently (not by brute force) all $ NP $ problems are solvable efficiently (polynomial cost). 

\begin{table}[h]
	\begin{center}
		\begin{tabular}{ >{\centering\arraybackslash} m{4cm}| >{\centering\arraybackslash} m{5cm} }
			\toprule
			TSP Graph & TSP Matrix\\
			\midrule
			
				\includegraphics[width=1\linewidth]{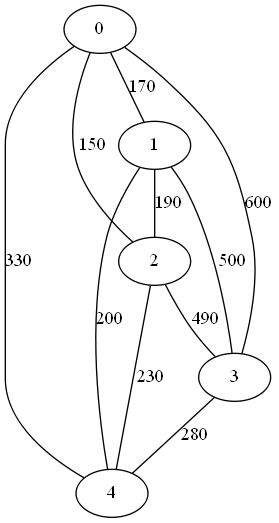}
				
			&
				
				\bordermatrix{ & 0 & 1 & 2 & 3 & 4 \cr 
					0 & 0 & 170 & 150 & 600 & 330  \cr
					1 & 170 & 0 & 190 & 500 & 200  \cr
					2 & 150 & 190 & 0 & 490 & 230  \cr 
					3 & 600 & 500 & 490 & 0 & 280  \cr
					4 & 330 & 200 & 230 & 280 & 0  \cr
				}
				
			\\
						\bottomrule
		\end{tabular}
	\end{center}
	\caption{\small \sl Cities distance table in $ km $ (a graph and matrix descriptions).}
	\label{tab:cityTable}
\end{table}

\begin{figure}[h]
	\centering
	\includegraphics[width=0.8\linewidth]{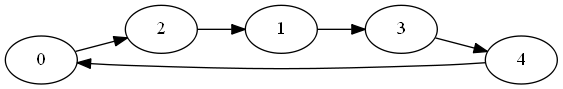}
	\caption{Permutation of a travel path $\sigma=\protect\begin{pmatrix} 0&1&2&3&4&0 \protect\\0&2&1&3&4&0 \protect\end{pmatrix}$.}
	\label{fig:travel}
\end{figure}

TSP starts with list of cities and their distances (Table~\ref{tab:cityTable}), where cities are enumerated as $ 0,1,2,3,4 $. 

The path is the list of cities, one permutation of that list. The number of permutations $ p $  is permutations with no repetitions $ p=n! $ where $ n $ is a number of cities. One path permutation is shown in Figure~\ref{fig:travel}.

Given the distances and a specific path, the distance can be calculated by Figure~\ref{fig:TSPfunct} program.

\begin{figure}[h]
	\begin{lstlisting}[style=CStyle]
	int getTravelingDistance(int *2dimArr, int *tsPath, int numberOfcities)
	{
		int i, distance = 0;
	
		for (i = 0; i < numberOfCities; i++)
			distance+=2dimArr[tsPath[i]][tsPath[i+1]];
	
		return distance;
	}
	\end{lstlisting}
	\caption{\small \sl gTD (getTravelingDistance) function. Given the number of cities, their distances and a path program returns path length.}
	\label{fig:TSPfunct}
\end{figure}

From our Table~\ref{tab:cityTable} and path from Figure~\ref{fig:travel} the arguments for function getTravelingDistance are:
\begin{itemize}
	\item $ 2dimArr $; Matrix Table~\ref{tab:cityTable}
	\item $ tsPath = \{0,2,1,3,4,0\} $ (indexes of visited cities Figure~\ref{fig:travel})
	\item $ numberOfcities=5 $
\end{itemize}
The function output is $ distance=150+190+500+280+330 $

\subsection{Rank 1}\label{tsp_rank1}
gTD (Figure~\ref{fig:TSPfunct}) satisfies Rank 1 requirements, the same as the $ sin $ program case (Figure~\ref{lst:sine}). It defines set of inputs, set of outputs and every individual member's mappings. 

However gTD size $ |d(s)| $ and the list of all input-output mappings $ |s| $ are in different relation than $ sin $ (Equation~\ref{eq:description1})
\begin{equation}\label{tsp_kol}
|d(s)|<<|s|
\end{equation}
because number of the possible travel paths permutations (inputs) rises exponentially $ |s|=n! $ (where $ n $ is number of cities) while gTD size $ |d(s)| $ remains constant.

\subsection{Rank 2} 
If Rank 1 exists, then Rank 2 is given because mappings can be pre-calculated and evaluated via a look-up table. In the travelling distance case, that approach is impractical because of the exponential table dependence on input.

\subsection{Rank 3 and 4} 
There is a similar connection between Rank 3 and Rank 4 as is in Rank 1-2. If a function has some mathematical structure, there is a good chance to find an inverse function as well (for example, $ sin $ and $ sin^{-1} $).

$ sin $ Rank 3 depends on the ordered structure of input-output pairs. Regarding gTD sorting, there is an immediate problem because gTD path instances do not have clear criteria of how to do a meaningful sort on path permutations. Table~\ref{tab:gTDmapping} shows a couple of gTD calculated distances sorted by the distance in ascending order.

\begin{table}[h]
	\begin{center}
		\begin{tabular}{cc}
			\toprule 
			path permutation&distance  \\ 
			\midrule 
			$0,2,1,3,4,0$ & $1450$  \\ 
			 
			$0,1,2,3,4,0$ & $1460$  \\ 
			
			$2,3,0,1,4,2$ & $1690$  \\ 
			  
			$3,2,1,4,0,3$ & $1810$  \\

		\end{tabular} 
	\end{center}
\caption{\small \sl gTD partial input-output mapping.}
\label{tab:gTDmapping}
\end{table}

It starts to resemble the random oracle Table~\ref{tab:randomOracle} where only one column is sorted. 
%If some meaningful sorting of gTD by path permutation exists, then the Table could be sorted by paths, and the corresponding distances will have a random appearance. 
There is a possibility that path permutation column (Table~\ref{tab:gTDmapping}) has some meaningful distribution after all. That eventuality leads to solving all $ NP $ complete problems efficiently ($O(logn)$), which is believed unlikely (\cite{gasarch2002p}~and~\cite{rosenberger2012p}). 

The essence of Rank 3 approach is to search for mappings in the Table as it happens in Rank 2 but not creating the complete Table. Instead, it calculates only needed mappings according to the table pattern. In gTD case, if we need to know short distances and their paths, search for them will calculate permutations in the top table region. Since search algorithm cost is in $ O(log n) $ range, despite the exponential amount of inputs ($ n! $), the overall cost is still polynomial. 

The Rank 4 approach for the same sort of questions will match particular permutation with desired distance and again as in Rank 3 solve TSP.

\section{Rank 3 Discussion}

Straight forward candidate for a computer program having only Rank 1 and 2 ability is Blum Blum Shub (BBS). It is a secure pseudo-random generator, and it is interesting because its security is tied to the hard problem (in this case, factoring \cite{blum1986simple}). Practically, if a pseudo-random stream (an output) created by BBS is presented, then finding matching seed (an input) leaves you with either exhaustive search or factoring large composite number which is known to be a hard problem ($ NP $).

Since factoring is not NP-complete problem and also there is a quantum computer algorithm for solving factoring in $ P $ \cite{shor1994algorithms}, stronger evidence for ranking categorisation is needed. The possible next step is to show how NP-complete problems relate to ranking categorisation. We can use TSP to show if it has Rank $ 3 $ implementation, then searching for the mappings of interest will be very efficient ($ O(log_2n) $).

The Branching Theorem could be considered as an even stronger argument for ranking categorisation. The argumentation starting point is Cyclomatic Complexity (CYC)~\cite{mccabe1976complexity}.

\subsection{Preliminaries}

Every source code could have a decision point where a program can decide which way execution might go. The Table~\ref{tab:CYC} first column is an example of a program with one branching where vertex $ 2 $ is that decision point. Depending on the program arguments, the execution of it can happen in two ways (second and third column). With added edge (dotted), the control graphs can be discussed in cyclic graph terms~\cite{pemmaraju2003cycles}. A graph with $n$ vertices and a single cycle going through them all is called a cycle graph $C_n$~\cite{pemmaraju2003cycles}. The control graph from Table~\ref{tab:CYC} have two-cycle graphs (Cycle 1 and Cycle 2). McCabe showed that every decision point in the graph (program flow control) doubles the number of cycle graphs~\cite{mccabe1976complexity}.

To simplify definitions, the source code and its flow control graph in our discussion will be a programming function $ pf $.

\theoremstyle{definition}
\begin{definition}{\textit{Programming function ($pf$).}}\\
	$ pf $ is a process which performs a some desired  task. It accepts an input and returns an output ($ y=pf(x) $), thus mimicking mathematical notion of function $ y=f(x) $ where $ x $ is input and $ y $ is output. An example of $ pf $ is the C language $ sin $ implementations (Figure~\ref{lst:sine}). The same is valid for $ getTravelingDistance $ function (Figure~\ref{fig:TSPfunct}). Both have set of arguments which are input to the function
	\[ \textit{float sin(\hl{float x})} \]
	and return value which is actually output.
	\[ \textit{\hl{float} sin(float x)} \] 
	Another $ pf $ requirement is: the function execution flow has to have a single entry and exit point. It enables the treatment of $ pf $ in graph cycle terms. Examples of flow controls can be seen in Table~\ref{tab:CYC} where the blue vertex means entry point and red one exit point. 
\end{definition}

\theoremstyle{definition}
\begin{definition}{\textit{Redirected programming function ($rf$).}}\\
	$ rf $ is essentially $ pf $ and differs from $ pf $ in only one detail. Instead outputting prescribed output as $ pf $ would do a $ rf $ will return a execution path of $ pf $ (a cycle). For example, co-domain $ Y $ of the now redirected program ($ rf $) from Table~\ref{tab:CYC} will be $Y=\{Cycle 1,\, Cycle 2\}$. Note that some functions as our $ sin $ implementations (Figure~\ref{lst:sine})    have only one output cycle (execution path).
\end{definition}

\begin{table}[h]
	\begin{center}
		\begin{tabular}{ >{\centering\arraybackslash} m{3.5cm} >{\centering\arraybackslash} m{3.5cm} >{\centering\arraybackslash} m{3.5cm}}
			\toprule
			
			Flow control graph & Cycle 1 & Cycle2\\
			
			\midrule
			
			\includegraphics[width=0.6\linewidth]{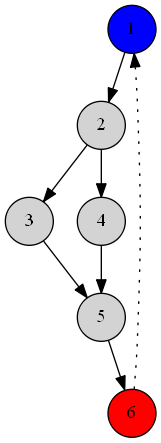}&
			\includegraphics[width=0.33\linewidth]{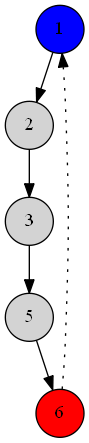}&
			\includegraphics[width=0.33\linewidth]{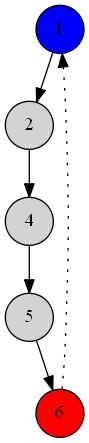} \\
			
			\midrule
			
			\bordermatrix{ & 1 & 2 & 3 & 4 & 5 & 6 \cr 
				1 & 0 & 1 & 0 & 0 & 0 & 0 \cr
				2 & 0 & 0 & 1 & 1 & 0 & 0 \cr
				3 & 0 & 0 & 0 & 0 & 1 & 0 \cr 
				4 & 0 & 0 & 0 & 0 & 1 & 0 \cr
				5 & 0 & 0 & 0 & 0 & 0 & 1 \cr
				6 & \textcolor{gray}{1} & 0 & 0 & 0 & 0 & 0 \cr}&
			
			\bordermatrix{ & 1 & 2 & 3 & 5 & 6 \cr 
				1 & 0 & 1 & 0 & 0 & 0 \cr
				2 & 0 & 0 & 1 & 0 & 0 \cr
				3 & 0 & 0 & 0 & 1 & 0 \cr 
				5 & 0 & 0 & 0 & 0 & 1 \cr
				6 & \textcolor{gray}{1} & 0 & 0 & 0 & 0 \cr}&
			
			\bordermatrix{ & 1 & 2 & 4 & 5 & 6 \cr 
				1 & 0 & 1 & 0 & 0 & 0 \cr
				2 & 0 & 0 & 1 & 0 & 0 \cr
				4 & 0 & 0 & 0 & 1 & 0 \cr
				5 & 0 & 0 & 0 & 0 & 1 \cr
				6 & \textcolor{gray}{1} & 0 & 0 & 0 & 0 \cr}\\
			
			\bottomrule
				
		\end{tabular} 
	\end{center}
	\caption{\small \sl Flow control graph and its cycle decomposition. Edge addition (dotted) makes graphs cyclic.}
	\label{tab:CYC}
\end{table}

\subsection{Branching Argument}

The $rf$ redirected functions share the same domain $ X $ and the same function description (plus output redirection) with the programming functions $ pf $. The cardinality of $RF$ (set of all $ rf $) is exactly the same as of $PF$ (set of all $ rf $) . 
\begin{equation}
|RF|=|PF|
\end{equation}

Now we can ask what is $Ranking\,3$ status of $rf:X \rightarrow Y$ ?

\begin{theorem}{\textit{Branching Theorem.}}\label{th:branch}
	There exists at least one $rf:X \rightarrow Y$ mapping not satisfying $Ranking\,3$ requirement.
\end{theorem}

\begin{proof}
	We can assume opposite, that every $rf:X \rightarrow Y$ has a $ Ranking\, 3 $ attribute. In that case, a mapping search from any cycle (output) from some flow control graph, to the corresponding input is efficient one ($ O(log_2n) $). On first glance, that is a reasonable statement because when $ pf $ is properly specified every case (cycle) behaviour should be fully defined. On the other hand, there are at least two problems with all $ rf $ having $ Ranking\, 3 $  status:
	\begin{itemize}
		\item $ pf $ are not always made with a purpose. The program could be written with arbitrarily many branching statements positioned in a program everywhere. If it compiles, it will run, but output behaviour will be undefined. 
		
		$ 3x+1 $ iterations~\cite{lagarias19853} can illustrate (not prove) the point. Take any positive integer $ x $ and if it is even divide it with $ 2 $ else multiply it with $ 3 $ and add $ 1 $. Iterate this procedure until the outputs start to repeat itself~\footnote{we do not even know if it will stop on every input!}.  One iteration step formula is:
		
		\begin{equation}\label{form:3x+1}
		f(x) = \left\{\begin{array}{lr} x/2 & if\quad x\equiv 0 \\
		3x+1 & if\quad x\equiv 1
		\end{array}\right. (mod\enskip2)
		\end{equation}
		
		To paraphrase Lagarias~\cite{3x+1}, $ 3x+1 $ transformation appears to be without any mathematical structure and every arbitrarily chosen iteration (branching decision) behave like a fairly flipped coin.
		
		\item If every $ rf $ has $Ranking\,3$ implementation then branching algorithmic structure is redundant. We will know the partition of inputs and their corresponding execution paths cycles, because there is some structure, which allows binary searches, in every $ rf $ mapping. Practically, we will know which path to execute for some input in advance without the need to test the branching statement. Therefore, we can write any $ rf $ (and consequently any $ pf $) without $ if/else $ alike statements. That situation will make programming endeavour very interesting because branching is one of the algorithm essentials. Note, B{\"o}hm-Jacopini theorem~\cite{bohm1966flow} shows that the operations and structure of only three statements (sequence, \textbf{\textit{selection}} and iteration) is a Turing-equivalent system of computation.
	\end{itemize}
	
\end{proof}

\section{Conclusions}

\begin{table}[h]
	\begin{center}
		\begin{tabular}{cccc|c}
			
			\toprule 
			
			Rankings & \textit{Rank 1} & \textit{Rank 2} & \textit{Rank 3} & \textit{Rank 4} \\ 
			\midrule
			Description & initial & initial & initial & inverse\\
			\midrule 
			Direction & forward & forward and inverse  & inverse & inverse \\ 
			
			\midrule
			 
			Category 1 & \cmark & \cmark & \cmark & \cmark  \\ 
			 
			Category 2 & \cmark & \cmark & \xmark & \xmark  \\
			
			\bottomrule
			
		\end{tabular}
	\end{center} 
	\caption{\small \sl Ranking attributes; The function and its description are used for the first three rankings. \textit{Rank 4} uses the inverse function of the initial one. The direction forward and inverse indicates evaluation way (given $ x $ or $ y $). On end, we have two categories of functions.}
	\label{tab:category} 
\end{table}

For the discussion \textit{Rank 4} is not essential. We could assume that if some function does not have \textit{Rank 3} it probably does not have \textit{Rank 4}. If the Theorem~\ref{th:branch} holds then we have \textit{Category 2} ($ C2 $) type of functions.

C2 functions have only one option to do the inverse, and that is by look-up table (\textit{Rank 2}). In the case when domain $ X $ is large, making pre-computation prohibitive, the only option left to evaluate function is to calculate a single instance of it. The same phenomena as C2 function called computational irreducibility were investigated by~\cite{wolfram2002new}~and~\cite{israeli2004computational}. The quote is from Computational Irreducibility~\cite{rowlandirreducibility} page:
\begin{quotation}
	The principle of computational irreducibility says that the only way to determine the answer to a computationally irreducible question is to perform, or simulate, the computation.  
\end{quotation}

The computational irreducibly was often associated with cellular automaton. The chaos theory provides another source of C2 functions. For example, $ n $ body problem does not have a practical mathematical solution. The only way to predict the future configuration after an extended time of moving bodies is to do simulation.

In practice, whenever the C2 functions are evaluated, the outcome shall be unexpected. There is a line of argumentation that puts C2 functions and random oracle (RO) in the same basket. C2 has capabilities as RO because it will produce an identical answer (output $ y $) from the same inquires (input $ x $). Every new inquiry, C2 will answer with surprising output $ y $. If $ y $ is somehow expected, then the particular function will have \textit{Rank 3} implementation and in some way ordered outputs. Essentially, \textit{Rank 3} means input/output correlation and C1. No \textit{Rank 3} means randomness and C2.   

\subsection{Kolmogorov Complexity; Additional Definitions}
As we saw in previous sections, some problems such as TSP are straightforward to state indicating low complexity; please see Subsection~\ref{tsp_rank1} and Equation~\ref{tsp_kol}. On the other hand, TSP mapping and its behaviour are undefined. To find some details such as the shortest path is believed to be hard and only exhaustive search method exists. 

The way around this discrepancy is to treat the TSP problem as a partially defined function. The fully described program should include a look-up table description as well. Table~\ref{tab:randomOracle} is an example where entries are sorted by input and output. In that way, an algorithm could find the required details efficiently ($ O(log_2n) $). In the TSP case, table size will rise exponentially with many cities impacting the size of the program and its complexity. Therefore, program size $ d(s) $ and look-up table size $ T $ are approximately the same as the size of all mappings $ |s| $ .
\begin{equation}
|d(s)|+|T|\approx|s|
\end{equation}
In that way, Kolmogorov Complexity will be high for TSP as should be (because TSP is NP-complete problem).

\subsection{Random Oracle Adjustment}
In similar vain as Kolmogorov Complexity, RO concept needs adjustment. In RO, full function description plays the main role. Negative result is argued by ratio of function description $ |d(s)| $ and size of the mappings $ |s| $.
\begin{equation}
|d(s)|<<|s|
\end{equation}
where shorter program must define behaviour of large function mappings size (in cryptography, modern hash output size is $ 256-512 $ bits).

That assertion does not work for partially defined functions( $ C2 $) which allows evaluation of single input without the need of whole mapping behaviours via exponentially sized look-up tables. This phenomenon enables trivial construction of one-way function~\cite{vuckovac2018one} (easy evaluation and exhaustive search for inverting).  

\bibliographystyle{unsrt}
\bibliography{sas}

\end{document}